\documentclass[12pt,technote, onecolumn, draft]{IEEEtran}
\usepackage[utf8]{inputenc}
\usepackage{latexsym}
\usepackage{mathrsfs}
\usepackage{graphicx}
\usepackage{multirow}
\usepackage{amsfonts,amssymb,amsmath,amsthm,bm}
\usepackage{color}
\usepackage{booktabs}
\usepackage{cite}
\newtheorem{theorem}{Theorem}[section] 
\newtheorem{definition}[theorem]{Definition} 
\newtheorem{lemma}[theorem]{Lemma} 
\newtheorem{corollary}[theorem]{Corollary}

\newtheorem{proposition}[theorem]{Proposition}

\begin{document}
	\title{A New Family of Binary Sequences via Elliptic Function Fields over Finite Fields of Odd Characteristics\thanks{Xiaofeng Liu and Fang-Wei Fu were supported by the National Key Research and Development Program of China (Grant Nos.  2022YFA1005000), the National Natural Science Foundation of China (Grant Nos. 12141108, 61971243), the Fundamental Research Funds for the Central Universities of China (Nankai University), and the Nankai Zhide Foundation. Jun Zhang was supported by the National Natural Science Foundation of China under Grant No. 12441105.}}
	\author{Xiaofeng Liu, Jun Zhang, Fang-Wei Fu
		\IEEEcompsocitemizethanks{\IEEEcompsocthanksitem Xiaofeng Liu and Fang-Wei Fu are with the Chern Institute of Mathematics and LPMC, Nankai University, Tianjin 300071, China, Emails: lxfhah@mail.nankai.edu.cn, fwfu@nankai.edu.cn. Jun Zhang is with the School of Mathematical Sciences, Capital Normal University, Beijing 100048, China, Email: junz@cnu.edu.cn.
		}
	}{\tiny }
    \maketitle

	\begin{abstract}
		 Motivated by the construction of binary sequences by utilizing the cyclic elliptic function fields over the finite field $\mathbb{F}_{2^{n}}$ by Jin \textit{et al.} in [IEEE Trans. Inf. Theory 71(8), 2025], we extend the construction to the cyclic elliptic function fields with odd characteristics by using the quadratic residue map $\eta$ instead of the trace map used therein. For any cyclic elliptic function field with $q+1+t$ rational points and any positive integer $d$ with $\gcd(d, q+1+t)=1$, we construct a new family of binary sequences of length $q+1+t$, size $q^{d-1}-1$, balance upper bounded by $(d+1)\cdot\lfloor2\sqrt{q}\rfloor+|t|+d,$ the correlation upper bounded by $(2d+1)\cdot\lfloor2\sqrt{q}\rfloor+|t|+2d$ and the linear complexity lower bounded by  $\frac{q+1+2t-d-(d+1)\cdot\lfloor2\sqrt{q}\rfloor}{d+d\cdot\lfloor2\sqrt{q}\rfloor}$, where $\lfloor x\rfloor$ stands for the integer part of $x\in\mathbb{R}$. 
         \end{abstract}

	
	\section{Introduction}
	\label{sec:1}
  Binary sequences with low correlations (auto-correlation and cross-correlation) and large linear complexities are widely applied in digital communication, cryptography, and signal processing. Several renowned constructions of binary sequences with good parameters have been obtained via different methods over the past decades. The oldest family of binary sequences, known as Gold sequence, was constructed in 1967. Such family of sequences has length $2^{n}-1$, size $2^{n}+1$ for odd integer $n$. The construction methods of Gold sequences can be found in \cite{13,14}. Gong constructed a family of binary sequences with size $2^{n}-1$, length $(2^{n}-1)^{2}$ and correlation bounded by $3+2(2^{n}-1)$ in 2002; see \cite{16}. Zhou and Tang obtained a family of binary sequences with length $2^{n}-1$ for $n=2m+1$, size $2^{\ell n}+\cdots+2^{n}-1$ and correlation $2^{m+\ell}+1$ for each $1\leq \ell\leq m$; see \cite{15}. Trace Norm sequences with good correlation properties were proposed in 1995; see \cite{27}.
 Binary sequences with lengths $2(2^{n}-1)$ and $(2^{n}-1)^{2}$ were constructed in 2007; see \cite{17}. The cyclic structure of the multiplicative group $\mathbb{F}_{2^{n}}^{*}$ was also applied in the design of binary sequences; see \cite{16,18} \textit{etc}. Several constructions of binary sequences over finite fields with odd characteristics were also proposed. In 2006, a family of binary sequences, now known as Weil sequences, with length an odd prime $p$, family size $(p-1)/2$, and correlation bounded by $5+2\sqrt{p}$ was constructed; see \cite{6}. Paterson proposed a family of pseudorandom binary sequences with length $p^{2}, p\equiv 3$ (mod $4$) and $p$ a prime; see \cite{29}. The construction employed Hadamard difference sets and maximum distance separable (MDS) codes.
  In \cite{23}, Jin \textit{et al.} constructed four families of binary sequences with flexibility on length while still possessing low correlation. 
  \begin{table}[htbp]  
  \centering  
   \caption{Parameters of sequence families}  
   \begin{tabular}{|c|c|c|c|}  
    \toprule
    Sequence  & Length $N$ & Family Size& Bound of Correlation \\    \midrule
   Gold (odd)(\cite{13})  & $2^{n}-1$, $n$ odd& $N+2$&$1+2\sqrt{2}\sqrt{N+1}$\\\midrule
  Gold (even)(\cite{13})  & $2^{n}-1$, $n=4k+2$& $N+2$&$1+2\sqrt{N+1}$\\   \midrule
  Kasami (small)(\cite{19})  & $2^{n}-1$, $n$ even & $\sqrt{N+1}$&$1+\sqrt{N+1}$\\ \midrule 
   Kasami (large)(\cite{19})  & $2^{n}-1$, $n=4k+2$ & $(N+2)\sqrt{N+1}$&$1+2\sqrt{N+1}$\\   \midrule
  Gong (\cite{16})  & $(2^{n}-1)^{2}$, $2^{n}-1$ prime & $\sqrt{N}$&$3+2\sqrt{N}$\\  \midrule    
 Paterson (\cite{29})  & $p^{2}$, $p$ prime, $p\equiv 3$ (mod) $4$ & $N$&$5+4\sqrt{N}$\\  \midrule  
 Paterson (\cite{29})  & $p^{2}$, $p$ prime, $p\equiv 3$ (mod) $4$ & $\sqrt{N}+1$&$3+2\sqrt{N}$\\  \midrule  
  
  Tang \textit{et al}. (\cite{17})& $2(2^{n}-1)$, $n$ odd & $\frac{N}{2}+1$&$2+\sqrt{N+2}$\\       \midrule
J. Rushanan \textit{et al}. (\cite{6})& $p$, prime & $\frac{N-1}{2}$&$5+2\sqrt{N}$\\       \midrule
 Hu \textit{et.al}. (\cite{2})&$p^{n}+1$, $p$ odd prime& $N-2$& $4+2\sqrt{N-1}$\\\midrule
 Jin \textit{et al}. (\cite{9})& $2^{n}+1$ &$N-2$&$2\sqrt{N-1}$\\\midrule
 Jin \textit{et al}. (\cite{23})& $q-1, q\geq 17$ prime power&$N+3$&$6+2\sqrt{N+1}$\\\midrule
Jin \textit{et al}.(\cite{23})& $q-1,q\geq 11$&$N/2$&$2+\sqrt{N+1}$\\\midrule
Jin \textit{et al}.(\cite{23})& $q,q\geq 17$ prime&$N$&$5+2\sqrt{N}$\\\midrule
Jin \textit{et al}.(\cite{23})& $q,q\geq 11$ prime&$(N-1)/2$&$1+2\sqrt{N}$\\\midrule
  Jin \textit{et al}. (\cite{3})& $2^{n}$&$N^{2}-1$&$1+14\sqrt{N}$\\\midrule
Jin \textit{et al}.(\cite{3})& $q+1\pm\sqrt{q},q=4^{n}$&$q-1$&$11\sqrt{q}$\\\midrule
Jin \textit{et al}.(\cite{3})& $q+1\pm\sqrt{2q},q=2\times4^{n}$&$q-1$&$(10+\sqrt{2})\sqrt{q}$\\\midrule
$\mathbf{Our}$ $\mathbf{construction}$& $q$, $q=5^{n}, n$ even&$N^{2}-1$&$7+14\sqrt{q}$\\\midrule
 $\mathbf{Our}$ $\mathbf{construction}$& $q+1\pm \sqrt{3q}, q=3^{n}$, $n$ odd &$q-1$&$4+(10+\sqrt{3})\sqrt{q}$\\\midrule
 $\mathbf{Our}$ $\mathbf{construction}$& $q+1\pm\sqrt{q},q=3^{n}$, $n$ even &$q-1$&$4+11\sqrt{q}$\\\midrule
 $\mathbf{Our}$ $\mathbf{construction}$& $q+1\pm\sqrt{q},q=5^{n}$, $n$ even &$q-1$&$4+11\sqrt{q}$\\\midrule
   \bottomrule
  \end{tabular}
  \label{9}  
\end{table}

The theory of function fields has been widely applied in prior constructions of sequences.  In \cite{2} (resp. \cite{9}), a family of binary sequences of length $2^{n}+1$ ( resp. $p^{n}+1$), size $2^{n}-1$ (resp. $p^{n}-1$) and correlation bounded by $2^{(n+2)/2}$ (resp. $4+\lfloor2\cdot p^{n/2}\rfloor$) was constructed via cyclotomic function fields by Jin \textit{et al.} (resp. Hu \textit{et al.}).  In \cite{5}, Xing \textit{et al.} gave a general construction of binary sequences with low-correlation, large linear span via function fields over finite fields with even characteristic. In \cite{10}, Hu \textit{et al.} extended the results of Xing \textit{et al.} by utilizing the theory of Kummer extensions. Luo \textit{et al.} constructed sequences with large nonlinear complexities from rational function fields and cyclotomic function fields; see \cite{20}. Hu \textit{et al.} studied a class of pseudorandom sequences from elliptic function fields over finite fields; see \cite{21}.  Recently, Jin \textit{et al.} constructed a novel family of binary sequences by using rational points on the cyclic elliptic curves over finite fields with even characteristic; see \cite{3}.

Binary sequences can also be constructed over finite fields with odd characteristics. The prior constructions mainly utilized the quadratic characters and the cyclic structure of finite fields such as \cite{2,23,6} \textit{etc.}. It is notable that the lengths of sequences in some constructions are limited by the size of finite fields (see table \ref{9}) or have  strict requirements for the characteristics, such as $p\equiv 3$ (mod $4$) in \cite{29}, $p\geq 17 $ in \cite{23}. In this paper, we construct a new family of sequences with the cyclic elliptic function fields over finite fields with odd characteristics by using the theory of Kummer extensions and the quadratic residue map, which extends the construction by Jin \textit{et al.} in \cite{3}. Then we obtain a new family of binary sequences and we also give the upper bound for the balance, correlation and the lower bound for the linear complexity. 
  
  By contrast, our construction can be applied in cyclic elliptic function fields over any finite fields with odd characteristics. In applications, the insertion or deletion of components of sequences generally destroys the favorable correlation property. Therefore, it is both theoretically and practically important to construct binary sequences with low correlation but not limited by the field sizes. It turns out that our family of binary sequences achieves a longer length controlled by the Serre bound, which means we provide a construction not constrained by the size of finite fields. On the other hand, the family size of our construction is more flexible for the fact that the valuation functions are taken from a specified Riemann-Roch space corresponding to a unique place with arbitrary degree.

   In Table \ref{9}, we list the parameters of several binary sequence families for comparison.

\subsection{Our Main Results and Techniques}
 In this subsection, we provide the main results of this paper and also give a brief introduction of our techniques.

Let $E/\mathbb{F}_{q}$ be an elliptic function fields with the constant field $\mathbb{F}_{q}$. Let $E(\mathbb{F}_{q})$ be the set of the $\mathbb{F}_{q}-$rational places of $E/\mathbb{F}_{q}$. It is well-known that $E(\mathbb{F}_{q})$ forms an abelian group under a specified addition law; see \cite{3} \textit{etc.}. In particular, if $E(\mathbb{F}_{q})$ forms a cyclic group, $E/\mathbb{F}_{q}$ is called a cyclic elliptic function field. The details of definition and properties of the cyclic elliptic function fields shall be introduced in the next section. 

Now we summarize the main results of this paper into the following theorem.

\begin{theorem}\label{mainresults}
     Let $q=p^{n}$ for an odd prime $p$ and a positive integer $n$. Let $E/\mathbb{F}_{q}$ be a cyclic elliptic function field defined over $\mathbb{F}_{q}$ with $q+1+t$ rational places. Let $Q$ be a place of $E$ with $\deg(Q)=d\geq 2$ such that $\gcd(d,q+1+t)=1$. Then we can construct a family of binary sequences $\mathcal{S}=\{\mathbf{s}_{i}: 1\leq i\leq q^{d-1}-1\}$ with length $q+1+t$, size $q^{d-1}-1$, balance bounded by
         \begin{displaymath}             
(d+1)\cdot\lfloor2\sqrt{q}\rfloor+|t|+d,
          \end{displaymath}
linear complexity lower bounded by 
$$\frac{q+1+2t-d-(d+1)\cdot\lfloor2\sqrt{q}\rfloor}{d+d\cdot\lfloor2\sqrt{q}\rfloor}$$
and correlation upper bounded by $$(2d+1)\cdot\lfloor2\sqrt{q}\rfloor+|t|+2d$$ where $\lfloor x\rfloor$ stands for the integer part of $x\in\mathbb{R}$.
          \end{theorem}

Although our construction and the construction in~\cite{3} are both based on the structures of rational places over cyclic elliptic function fields, the main tools in our construction are taken from the theory of Kummer extensions, whereas the main techniques used in \cite{3} are Artin-Schreier extension theory. The detailed comparison shall be provided in Section \ref{sec:3}.

 We give a remark on the main result to finish this subsection that by utilizing the computing software SageMath~\cite{sage}, the actual parameter values of the sequences are much smaller than the theoretical bounds given in Theorem~\ref{mainresults}. Please see Tables~\ref{tab:5} and~\ref{tab:4} in Section \ref{sec:3} for the comparison.

\subsection{Organization of This Paper}
   The rest of this paper is organized as follows. In Section \ref*{sec:2}, we provide preliminaries on binary sequences, extension theory of function fields, elliptic function fields and Kummer extensions of function fields. In Section \ref*{sec:3}, we present a theoretical construction of binary sequences with low correlation, favorable balance and large linear complexity via cyclic elliptic function fields over finite fields with odd characteristics by utilizing the theory of Kummer extensions. We also provide some examples to illustrate our construction.
   In Section \ref*{sec:4}, we give a conclusion and some possible future work on the constructions of binary sequences.
	
	\section{Preliminaries }
	\label{sec:2}
	In this section, we list some preliminaries on binary sequences, extension theory of function fields, elliptic function fields and Kummer extensions.
    \subsection{Autocorrelation and Cross Correlation of Binary Sequences}
	Let $N$ be a positive integer. A binary sequence of length $N$ is a vector $\mathbf{s}\in\mathbb{F}^{N}_{2}$.  For a binary sequence $\mathbf{s}=(s_{0},s_{1},\cdots,s_{N-1})$, the corresponding autocorrelation $A_{t}(\mathbf{s})$ at delay $t$ with $1\leq t\leq N-1$ is defined as 
    \begin{displaymath}
        A_{t}(\mathbf{s}):=\sum^{N-1}_{j=0}(-1)^{s_{j}+s_{j+t}},
    \end{displaymath}
	where $j+t$ is identified with the least non-negative integer modulo $N$. For two distinct binary sequences $\mathbf{u}=(u_{0},u_{1},\cdots,u_{N-1})$ and $\mathbf{v}=(v_{0}, v_{1},\cdots,v_{N-1})\in\mathbb{F}_{2}^{N}$, their cross correlation at delay $t$ with $0\leq t\leq N-1$ is defined as 
    \begin{displaymath}
        C_{t}(\mathbf{u,\mathbf{v}}):=\sum^{N-1}_{j=0}(-1)^{u_{j}+v_{j+t}}.
        \end{displaymath}
Let $\mathcal{S}$ be a family of binary sequences of length $N$.  Denote by $A(\mathcal{S})=\max\{|A_{t}(\mathbf{s})|:\mathbf{s}\in\mathcal{S}, 1 \leq t\leq N-1\}$ and $C(\mathcal{S})=\max\{|C_{t}(\mathbf{u,v})|:\mathbf{u}\neq\mathbf{v}\in\mathcal{S}, 0\leq t\leq N-1\}\}$,
then the correlation of $\mathcal{S}$ is defined as $Cor(\mathcal{S}):=\max\{A(\mathcal{S}), C(\mathcal{S})\}$.
\subsection{Linear Complexity of Binary Sequences}
Now we introduce the definition of linear complexity of periodic binary sequences. For any nonzero sequence $\mathbf{s}\in\mathcal{S}$, the corresponding linear complexity is the smallest positive integer $\ell=\ell(\mathbf{s})$ such that there exist $\ell+1$ elements $\lambda_{0},\lambda_{1},\cdots,\lambda_{\ell}\in\mathbb{F}_{2}$ with $\lambda_{0}=\lambda_{\ell}=1$ and
\begin{displaymath}
    \sum^{\ell}_{i=0}\lambda_{i}s_{i+u}=0
\end{displaymath}
 for any $u\geq 0$. We define the linear complexity of $\mathcal{S}$ by
\begin{displaymath}
    \mathrm{LC}(\mathcal{S}):=\min\{\ell(\mathbf{s}) : \mathbf{s}\in\mathcal{S}\}.
\end{displaymath}
\subsection{Extension Theory of Function Fields}
Let $F/\mathbb{F}_{q}$ be an algebraic function field with genus $\mathfrak{g}_{F}$ over the constant field $\mathbb{F}_{q}$. Any place of $F$ which is invariant under the action of Galois group $\mathrm{Gal}(\overline{\mathbb{F}}_{q}/\mathbb{F}_{q})$ is called rational. Denote by $N(F)$ the number of all rational places of $F$. The Serre bound improves the well-known Hasse-Weil bound and it is given by:
\begin{displaymath}
    |N(F)-q-1|\leq \mathfrak{g}_{F}\cdot\lfloor2\sqrt{q}\rfloor.
\end{displaymath}
Here $\lfloor x\rfloor$ stands for the integer part of $x\in\mathbb{R}$.

A divisor is a formal sum of places. Denote by $\mathbb{P}_{F}$ and $\mathbb{D}_{F}$ the set of places and the set of divisors respectively. Let $\mathrm{v}_{P}$ be a normalized discrete valuation of $F$. For divisor $G=\sum_{P\in\mathbb{P}_{F}}\mathrm{v}_{P}(G)P\in \mathbb{D}_{F}$, the support of $G$ is denoted by $\mathrm{Supp}(G)=\{P\in\mathbb{P}_{F}\ |\ \mathrm{v}_{P}(G)\neq0\}$. The degree of divisor $\deg(G)$ is defined as $\deg(G)=\sum_{P\in\mathbb{P}_{F}}\mathrm{v}_{P}(G)\deg(P)$. A divisor $G=\sum_{P\in\mathbb{P}_{F}}\mathrm{v}_{P}(G)P$ is called effective if $\mathrm{v}_{P}(G)\geq 0$ for all $P\in \mathbb{P}_{F}$.

For any rational function $f\in F\setminus\{0\}$, we have the principal divisor $(f)=\sum_{P\in\mathbb{P}_{F}}\mathrm{v}_{P}(f)P$ and it satisfies $\deg((f))=0$. The Riemann-Roch space associated to a non-negative divisor $G$ is given by the following:
$$\mathcal{L}(G):=\{f\in F\setminus\{0\}\ |\ (f)+G\geq 0\}\cup\{0\}$$
and it is a finite dimensional $\mathbb{F}_{q}-$vector space, and we denote by $\ell(G)$ the dimension of $\mathcal{L}(G)$. If $\deg(G)\geq 2\mathfrak{g}_{F}-1$, then $\ell(G)=\deg(G)+1-\mathfrak{g}_{F}$ from the well-known Riemann-Roch theorem; see \cite{1}.

Denote by $\mathrm{Aut}(F/\mathbb{F}_{q})$ the automorphism group of $F$ over $\mathbb{F}_{q}$, i.e.,
\begin{displaymath}
    \mathrm{Aut}(F/\mathbb{F}_{q})=\{\sigma : F\rightarrow F\mid \sigma\ \text{is an $\mathbb{F}_{q}-$automorphism of}\ F\}.
\end{displaymath}
In the following, we shall consider the group action of automorphism group 
$\mathrm{Aut}(F/\mathbb{F}_{q})$ on $\mathbb{P}_{F}$. Fix a place $P\in\mathbb{P}_{F}$. Then for any automorphism $\sigma\in \mathrm{Aut}(F/\mathbb{F}_{q})$, $\sigma(P)$ is also a place of $F$ and we have the following results.
\begin{lemma}(see \cite{7})
    For any automorphism $\sigma\in \mathrm{Aut}(F/\mathbb{F}_{q})$, $P\in\mathbb{P}_{F}$ and $f\in F$, we have
    \begin{enumerate}
        \item $\deg(\sigma(P))=\deg(P)$;
        \item $\mathrm{v}_{\sigma(P)}(\sigma(f))=\mathrm{v}_{P}(f)$;
        \item $\sigma(f)(\sigma(P))=f(P)$, provided that $\mathrm{v}_{P}(f)\geq 0$.
    \end{enumerate}
    \label{1}
    \end{lemma}
    Let $T/\mathbb{F}_{q}$ be a finite extension of $F/\mathbb{F}_{q}$. 
     For any $P\in\mathbb{P}_{F}$ and $Q\in\mathbb{P}_{T}$ with $Q|P$, we denote by $d(Q|P), e(Q|P), f(Q|P)$ the different exponent, ramification index and relative degree of $Q|P$, respectively. The different of $T/F$ is defined as $\mathrm{Diff}(T/F)=\sum_{P\in\mathbb{P}_{F}}\sum_{Q|P}d(Q|P)Q$. If $e(Q|P)$ is coprime to the characteristic $p$, then $d(Q|P)=e(Q|P)-1$ by Dedekind's Different Theorem. Otherwise, the different exponents can be determined by using Hilbert's Different Theorem and orders of the higher ramification groups; see~\cite{1}. Then the Hurwitz genus formula (see \cite{1})  is given by
    \begin{displaymath}
        2\mathfrak{g}_{T}-2=[T: F]\cdot (2\mathfrak{g}_{F}-2)+\deg(\mathrm{Diff}(T/F)).
    \end{displaymath}

    \subsection{Elliptic Function Fields}
An elliptic function field is an algebraic function field of genus one. To discuss isogeny classes, we use the geometric language, \textit{i.e.,} elliptic curves in Lemma \ref{22222}.
Two elliptic curves $E_{1}$ and $E_{2}$ are isogenous if there is a non-constant smooth $\mathbb{F}_{q}-$morphism that sends the zero of $E_{1}$ to the zero of $E_{2}$. It is well known that two elliptic curves over $\mathbb{F}_{q}$ are isogenous if and only if they share the same number of rational points; see \cite{4}. In particular, we denote them by $E_{1}\sim E_{2}$ if $E_{1}$ and $E_{2}$ are isogenous. Denote by $\mathcal{E}(\mathbb{F}_{q})$ the set of all elliptic curves with constant field $\mathbb{F}_{q}$. For any elliptic curve $E\in \mathcal{E}(\mathbb{F}_{q}) $, its isogeny class is defined by $[E]=\{E'\in\mathcal{E}(\mathbb{F}_{q})| E'\sim E\}$. 

Suppose $E/\mathbb{F}_{q}$ is an elliptic function field. Let  $E(\mathbb{F}_{q})$ denote the set of $\mathbb{F}_{q}-$rational places of $E/\mathbb{F}_{q}$. It is well-known that $E(\mathbb{F}_{q})$ forms an abelian group where we denote by $\oplus$ the group operation and denote by $O$ the infinity place; see \cite{11}. For any $P\in E(\mathbb{F}_{q})$ and $j\in\mathbb{Z}$, $[j]P$ means the addition of $j$ times of $P$, \textit{i.e.},
\begin{displaymath}
 \underbrace{P\oplus,\cdots,\oplus P}_{j}=\begin{cases}
      [j]P & \text{if}\ P\neq O,\\
       O&\text{if}\ P=O.
        \end{cases}
\end{displaymath}
The following Lemma classifies the existences of different elliptic function fields.
\begin{lemma}(see \cite{8})
    The isogeny classes of elliptic curves over $\mathbb{F}_{q}$ for $q=p^{n}$ are in one-to-one correspondence with the rational integers $t$ satisfying one of the following conditions:
    \begin{enumerate}
        \item $|t|\leq 2\sqrt{q}$ and $\gcd(t,p)=1$;
        \item $t=\pm2\sqrt{q}$ if $n$ is even;
        \item $t=\pm\sqrt{q}$ if $n$ is even and $p \not\equiv 1\mod 3$;
        \item $t=\pm p^{(n+1)/2}$ if $n$ is odd and $p=2$ or $3$;
        \item $t=0$ if either $1)$ $n$ is odd or $2)$ $n$ is even and $p\not\equiv 1\mod 4$.
    \end{enumerate}
    \label{22222}
\end{lemma}
For $t$ satisfying one of the conditions above, there exist elliptic function fields  with $q+1+t$ rational places in the isogeny class. The following lemma reveals the structure of the additive group of rational places on the elliptic function fields; see \cite{8}. In particular, we focus on the cyclic cases in this paper.
\begin{lemma}(see \cite{8})
    Let $\mathbb{F}_{q}$ be a finite field with $q=p^{n}$ elements. Let $h$ be a possible number of rational places of an elliptic function field $E$ of $\mathbb{F}_{q}$. Consider the decomposition in powers of prime numbers $h=\prod_{r}r^{h_{r}}$. Then the possible cyclic structures of $E(\mathbb{F}_{q})$ are given by
    \begin{displaymath}
        \mathbb{Z}/p^{h_{p}}\mathbb{Z}\times\prod_{r\neq p}( \mathbb{Z}/p^{a_{r}}\mathbb{Z}\times \mathbb{Z}/p^{h_{r}-a_{r}}\mathbb{Z})
        \end{displaymath}
        with
        \begin{enumerate}
        \item In cases  $3)$ and $4)$ of Lemma \ref{22222}: $a_{r}$ is an arbitrary integer satisfying $0\leq a_{r}\leq\min\{v_{r}(q-1),\lfloor h_{r}/2\rfloor\}$, we have $E(\mathbb{F}_{q})~\simeq\mathbb{Z}/h\mathbb{Z}$.
        \item In cases $5)$ of Lemma~\ref{22222}: if $q\not\equiv -1\mod 4$, then $E(\mathbb{F}_{q})\simeq\mathbb{Z}/(q+1)\mathbb{Z}.$
        \end{enumerate}       
        \label{123}
\end{lemma}

 The automorphism group of an elliptic function field $E$ over $\mathbb{F}_{q}$ is in the form of the semidirect product $T_{E}\rtimes\mathrm{Aut}(E,O)$; see \cite{4}, where $\mathrm{Aut}(E,O)$ consists of the $\mathbb{F}_{q}-$automorphisms of $E$ fixing the infinity place $O$, and $T_{E}$ corresponds to the translation group $E$.  The elements in $T_{E}$ can be given by $T_{E}=\{\sigma_{P} : P\in E(\mathbb{F}_{q})\}$ where $\sigma_{P}$ is defined by $\sigma_{P}(Q)=P\oplus Q$ for any rational places $P$ and $Q$ of $E$. Then we only need to consider the actions of $T_{E}$ on the places of $E$. 

Let $F$ be the fixed subfield of $E$  with respect to the translation group $T_{E}$, {i.e.},
    \begin{displaymath}
    F=E^{T_{E}}=\{z\in E : \tau(z)=z\ \text{for any}\ \tau\in T_{E}\},
\end{displaymath}
where $\tau(z)(P)=z(\tau^{-1}(P))$ for any $P\in E(\mathbb{F}_{q})$ and $\tau\in T_{E}$. From  Galois theory, $E/F$ is a Galois extension with $\mathrm{Gal}(E/F)\simeq T_{E}$. In particular, if $E(\mathbb{F}_{q})$ forms a cyclic group of order $q+1+t$ under the addition $\oplus$, then $E$ is called a cyclic elliptic function field.

 In the following, we assume that $t$ satisfies one of the three conditions given in Lemma \ref{123}. Hence, there exists a cyclic elliptic function field $E/\mathbb{F}_{q}$ with $q+1+t$ rational places. Then the translation group $T_{E}$ of $E$ is a cyclic group with the generator $\sigma_{P}$ and the order $q+1+t$. Let $P_{j}=\sigma_{P}^{j}(O)=[j]P$ be the rational place of $E$ for each $0\leq j\leq q+t$. Then the set $\{P,P_{1},\cdots,P_{q+t}\}$ consists of all the rational places of $E$.

Let $d$ be a positive integer relatively prime to $q+1+t$. Then the following Lemmas illustrate that there are exactly $r=B_{d}/(q+1+t)$ places of $F=E^{T_{E}}$ of degree $d$ which split completely in the extension $E/F$.

\begin{lemma}[see \cite{3}]
    Let $E/\mathbb{F}_{q}$ be an elliptic function field with the constant field $\mathbb{F}_{q}$ and $q+1+t$ rational places. Let $B_{d}$ be the number of places with degree $d$ of $E$. Let $\mu(\cdot)$ denote the M$\ddot{o}$bius function. Then we have
    \begin{displaymath}
        B_{d}=\frac{1}{d}\cdot\sum_{r\mid d}\mu\left(\frac{d}{r}\right)\cdot\left(q^{r}+1-S_{r}\right),
    \end{displaymath}
where 
\begin{displaymath}
S_{r}=\sum^{\lfloor r/2\rfloor}_{i=0}(-1)^{r-i}\frac{(r-i-1)!\cdot r}{(r-2i)!\cdot i!}t^{r-2i}q^{i}.
\end{displaymath}
In particular, we have $B_{2}=(q^{2}+q-t^{2}-t)/2$ and $B_{3}=(q^{3}-q+t^{3}-3qt-t)/3$.
\label{111}
\end{lemma}

\begin{lemma}[see \cite{3}]
Let $E/\mathbb{F}_{q}$ be a cyclic elliptic function field with the constant field $\mathbb{F}_{q}$ and $q+1+t$ rational places. Let $P$ be a generator of $E(\mathbb{F}_{q})$. Then for any place $Q$ in $E(\mathbb{F}_{q})$ of degree $d$ with $\gcd(d,q+1+t)=1$, and any fixed positive integer $i$, the places $\sigma_{P}^{i}(Q),\sigma_{P}^{i+1}(Q),\cdots,\sigma^{i+q+t}_{P}(Q)$ are pairwise different.
\label{l6}
\end{lemma}

In this paper, we shall utilize the quadratic residue map $\eta$ and this structure of rational places over the cyclic elliptic function fields with odd characteristics to construct a new family of binary sequences with excellent parameters.

\subsection{Kummer Extensions of Function Fields}
 In this subsection, we recall some properties of Kummer extensions of function fields which will be extensively used in the later discussions. For details, the readers may refer to \cite{11}. 
\begin{proposition}
    For a positive integer $m$, we let $F/K$ be an algebraic function field where $K$ contains a primitive $m$-th root of unity. Suppose that $u\in F$ is an element satisfying
    \begin{displaymath}
        u\neq w^{d}\ \ \text{for all}\ w\in F\ \text{and}\ d\mid m.
    \end{displaymath}
    Let
    \begin{displaymath}
        F'=F(y)\ \text{with}\ y^{m}=u.
    \end{displaymath}
    Such an extension $F'/F$ is said to be a Kummer extension of $F$. We have:
    \begin{enumerate}
        \item The polynomial $\Phi(T)=T^{m}-u$ is the minimal polynomial of $y$ over $F$ (in particular, it is irreducible over $F$). The extension $F'/F$ is Galois of degree $[F':F]=m$; its Galois group is cyclic, and the automorphisms of $F'/F$ are given by $\sigma(y)=\zeta y$, where $\zeta\in K$ is an $m-$th root of unity.
        \item Let $P\in\mathbb{P}_{F}$ and $P'\in\mathbb{P}_{F'}$ be an extension of $P$. Then 
        \begin{displaymath}
            e(P'|P)=\frac{m}{r_{P}}\ \text{and}\ d(P'|P)=\frac{m}{r_{P}}-1,
        \end{displaymath}
        where $r_{P}:=\gcd(m,\mathrm{v}_{P}(u))>0$ is the greatest common divisor of $m$ and $\mathrm{v}_{P}(u)$.
        \item Assume that there is $Q\in\mathbb{P}_{F}$ such that $\gcd(\mathrm{v}_{Q}(u),m)=1$. Then $K$ is the full constant field of $F'$, the extension $F'/F$ is cyclic of order $m$, and the genus of $F'$ satisfies
        \begin{displaymath}
2\mathfrak{g}_{F'}-2=m(2\mathfrak{g}_{F}-2)+\sum_{P\in\mathbb{P}_{F}}(m-r_{P})\deg P.
        \end{displaymath}
    \end{enumerate}
    \label{3}
\end{proposition}

	\section{New Binary sequences via Cyclic elliptic function fields with odd characteristics}
	\label{sec:3}
	Motivated by the constructions of binary sequences with low correlation via cyclotomic function fields and cyclic elliptic curves; see \cite{2, 9,3} respectively, we construct a novel family of binary sequences with favorable parameters via the cyclic elliptic function fields over finite fields with odd characteristics.
	\subsection{A New Construction of Binary Sequences via Cyclic Elliptic Function Fields}
     Let $\mathbb{F}_{q}$ be a finite field with $q$ an odd prime power. In this subsection, we shall utilize the algebraic structure of rational places on cyclic elliptic function field $E/\mathbb{F}_{q}$ and the quadratic residue map $\eta$ defined below to construct a new family of binary sequences with low correlation and large linear complexity.

    The quadratic residue map $\eta:\mathbb{F}_{q}\rightarrow\mathbb{F}_{2}$ with $q$ an odd prime power is defined by
\begin{displaymath}
  \eta(a)=\begin{cases}
        0,\ \text{if there exists}\ b\in\mathbb{F}^{*}_{q}\ \text{such that}\ a=b^{2}\ \text{or}\ a=0\\
        1,\ \text{otherwise}.
    \end{cases}
\end{displaymath}
It is clear that $\eta(a_{1}a_{2})=\eta(a_{1})+\eta(a_{2}),$ for any $a_{1},a_{2}\in\mathbb{F}_{q}^{*}$.

Let $Q$ be a place of cyclic elliptic function field $E$ with $\deg(Q)=d\geq 2$ such that $\gcd(d,q+1+t)=1$. By the Riemann-Roch theorem, we have $\ell(Q)=\dim(\mathcal{L}(Q))=\deg(Q)+1-\mathfrak{g}_{E}=\deg(Q)=d$ since $\deg(Q)=d\geq 2\mathfrak{g}_{E}-1=1$. Obviously, we have the decomposition $\mathcal{L}(Q)=\mathbb{F}_{q}\oplus V$. The vector space $V$ is called the complementary dual space of $\mathbb{F}_{q}$ in $\mathcal{L}(Q)$. The elements in $V\setminus\{0\}$ can be listed as $\{z_{1},z_{2},\cdots,z_{q^{d-1}-1}\}$.

Let $\sigma_{P}$ be a generator of the cyclic translation group $T_{E}$ for some rational place $P$ on $E(\mathbb{F}_{q})$. That is, $P_{j}=\sigma^{j}_{P}[O]=[j]P$ for $0\leq j\leq q+t$ are all rational places of $E$. For any function $z_{i}\in V\setminus\{0\}$ with $1\leq i\leq q^{d-1}-1$, the binary sequence $\mathbf{s}_{i}$ associated to $z_{i}$ is defined as follows:
\begin{displaymath}
    \mathbf{s}_{i}=(s_{i,0},s_{i,1},\cdots,s_{i,q+t})\ \text{with}\ s_{i,j}=\eta(z_{i}(P_{j}))\ \text{for}\ 0\leq j\leq q+t.
    \end{displaymath}
 Then the family of binary sequences $\mathcal{S}:=\{\mathbf{s}_{i}: 1\leq i\leq q^{d-1}-1\}$ has length $q+1+t$ and size $q^{d-1}-1$.  In the following, we shall illustrate that the family $\mathcal{S}$ has favorable balance, low correlation and large linear complexity.

    \subsection{The Balance and Correlation}
	In this subsection, we shall show that the family of binary sequences $\mathcal{S}=\{\mathbf{s}_{i} : 1\leq i\leq q^{d-1}-1\}$ has favorable balance and low correlation. 
    \begin{lemma}\label{2}
        For any elements $z_{1}, z_{2}\in V\setminus\{0\}$, we have 
        \begin{enumerate}
            \item $(z_{1}\sigma(z_{2}))_{\infty}=Q+\sigma(Q)$ for any automorphism $\sigma\in\mathrm{Gal}(E/F);$
            \item $\deg((z_{1}\sigma(z_{2}))_{\infty})=2d$ for any automorphism $\sigma\in\mathrm{Gal}(E/F).$
            \end{enumerate}
    \end{lemma}
        \begin{proof}
            For any functions $z_{1},z_{2}\in V\setminus\{0\}$, the pole divisors are given by $(z_{1})_{\infty}=(z_{2})_{\infty}=Q$. By Lemma \ref{1}, we have $\mathrm{v}_{\sigma(Q)}(\sigma(z_{2}))=\mathrm{v}_{Q}(z_{2})=-1$ and $(\sigma(z_{2}))_{\infty}=\sigma(Q)$. Then for any automorphism $\sigma\in\mathrm{Gal}(E/F)\setminus\{\mathrm{id}\}$, we have $z_{1}\sigma(z_{2})\in\mathcal{L}(Q+\sigma(Q))$ and  $\sigma(Q)\neq Q$. If $\sigma=\mathrm{id}$, then we have $(z_{1}\sigma(z_{2}))_{\infty}=(z_{1}z_{2})_{\infty}=2Q$. Therefore, $(z_{1}\sigma(z_{2}))_{\infty}=Q+\sigma(Q)$ for any automorphism $\sigma\in\mathrm{Gal}(E/F)$.
            
            For the second claim, for functions $z_{1},z_{2}\in V\setminus\{0\}$, we have $(z_{1}z_{2})_{\infty}=2Q$ since $(z_{1})_{\infty}=(z_{2})_{\infty}=Q$. By the relation $\deg(Q)=\deg(\sigma(Q))=d$ for $\mathrm{id}\neq \sigma\in\mathrm{Gal}(E/F)$, we have $\deg((z_{1}\sigma(z_{2}))_{\infty})=2d$ for any automorphism $\sigma\in\mathrm{Gal}(E/F)$.
            \end{proof}

The following Lemma determines all the ramified places and calculates the upper bound of the genus of the Kummer extension $E'/E$ over the elliptic function field $E/\mathbb{F}_{q}$.
    \begin{lemma}   
        Let $\mathbb{F}_{q}$ be a finite field with $q$ an odd prime power and $E/\mathbb{F}_{q}$ an elliptic function field. Consider the Kummer extension $E'=E(y)/E$ with defining equation $$y^{2}=z$$ for some function $z\in V\setminus\{0\}$. Then we have the following results: 
        \begin{enumerate}
        \item The ramified places in $\mathbb{P}_{E}$ are those $P\in\mathbb{P}_{E}$ with $\gcd(2,\mathrm{v}_{P}(z))=1$, \textit{i.e.}, $\mathrm{v}_{P}(z)$ must be odd. 
        \item The genus $\mathfrak{g}_{E'}$ of the Kummer extension $E'$ is upper bounded by
    \begin{equation}
        \mathfrak{g}_{E'}\leq 1+d.
        \label{3333}
        \end{equation}
    \end{enumerate}
    \end{lemma}
        \begin{proof}
        For the first claim, by the Proposition \ref{3}, if any place $P\in\mathbb{P}_{E}$ is ramified in $E'/E$ and $P'|P$ with $P'\in\mathbb{P}_{E'}$, it satisfies
            \begin{displaymath}
                e(P'|P)=\frac{2}{r_{P}}\ \text{and}\ d(P'|P)=\frac{2}{r_{P}}-1
            \end{displaymath}
where $r_{P}=\gcd(2,\mathrm{v}_{P}(z))$ and we have the following classifications.

If $r_{P}=\gcd(2,\mathrm{v}_{P}(z))=1$, then the ramification index and different exponent are given by 
\begin{displaymath}
                e(P'|P)=2\ \text{and}\ d(P'|P)=1.
            \end{displaymath}
If $r_{P}=\gcd(2,\mathrm{v}_{P}(z))=2$, then we have 
\begin{displaymath}
                e(P'|P)=1\ \text{and}\ d(P'|P)=0
            \end{displaymath}
which means $P\in\mathbb{P}_{E}$ is unramified in the Kummer extension $E'/E$. 

For the second claim, based on the analysis of behaviors of the ramified places, we have determined all the ramified places in $\mathbb{P}_{E}$ in the Kummer extension $E'/E$. According to the Hurwitz genus formula, we have the following estimation:
\begin{displaymath}
\begin{split}
    2\mathfrak{g}_{E'}-2&=[E':E](2\mathfrak{g}_{E}-2)+\deg\mathrm{Diff}(E'/E)\\
    &=\deg\mathrm{Diff}(E'/E)\ (\text{by the genus}\ \mathfrak{g}_{E}=1)\\
    &=\deg(Q)+\sum_{P\in\mathcal{Z}(z),\ r_{P}=1}\deg P\\
    &\leq d+\sum_{P\in\mathcal{Z}(z)}\deg P\\
    &\leq d+\sum_{P\in\mathcal{Z}(z)}\mathrm{v}_{P}(z)\deg P=2d.
    \end{split}
\end{displaymath}
where $\mathcal{Z}(z)$ denotes the set of the zero places of $z$. Then we have $2\mathfrak{g}_{E'}-2\leq 2d$, \textit{i.e.}, $\mathfrak{g}_{E'}\leq 1+d.$
    \end{proof}

\begin{lemma}[see \cite{5}]
    Let $z\in V\setminus\{0\}$. Put $P_{i}=\tau^{i}(P)$ for all $1\leq i\leq N-1$ and $\tau\in\mathrm{Aut}(E/\mathbb{F}_{q})$, then $P_{j+N}=P_{j}$ for all $j\in\mathbb{Z}$ and $P_{\ell},P_{\ell+1},\cdots,P_{\ell+N-1}$ are $n$ pairwise distinct rational places for any fixed $\ell\in\mathbb{Z}$ and we have $z(P_{\ell+u})=z(\tau^{u}(P_{\ell}))=\tau^{-u}(z)(P_{\ell})$ for any $0\leq u\leq N-1$.
    \end{lemma}
Based on the lemmas above, we have the following results of the correlation and the balance for this new family of binary sequences $\mathcal{S}=\{\mathbf{s}_{i}: 1\leq i\leq q^{d-1}-1\}$. It turns out that the estimations of the balance and correlation are transformed into the analysis of the number of rational places for two classes of Kummer extensions.

First, for cryptographic applications, a binary sequence must have favorable balance, \textit{i.e.}, the number of $1$s and $0$s are close. Now we give a formal definition for this property.
\begin{definition}
    Denoted by $T_{0}(\mathbf{s})$ and $T_{1}(\mathbf{s})$ the number of $0$s and $1$s in a binary sequence $\mathbf{s}$ respectively. The binary sequence $\mathbf{s}$ is called balanced if 
    \begin{displaymath}
|T_{1}(\mathbf{s})-T_{0}(\mathbf{s})|=|2T_{1}(\mathbf{s})-N|=|2T_{0}(\mathbf{s})-N|\leq 1
    \end{displaymath}
    where $N$ is the length of the sequence $\mathbf{s}$.
\end{definition}

Now we define $\Delta(\mathcal{S}):=\max\{|T_{1}(\mathbf{s})-T_{0}(\mathbf{s})|\ : \ \mathbf{s}\in\mathcal{S}\}$ and $\Delta(\mathcal{S})$ is called the balance for the family of sequences $\mathcal{S}$.  

 \begin{theorem}\label{balance}
         Let $q=p^{n}$ for an odd prime $p$ and a positive integer $n$. Let $E/\mathbb{F}_{q}$ be a cyclic elliptic function field defined over $\mathbb{F}_{q}$ with $q+1+t$ rational places. Let $F$ be the fixed subfield of $E$ with respect to its translation group $T_{E}$ generated by $\sigma_{P}$. Let $Q$ be a place of $E$ with $\deg(Q)=d\geq 2$ such that $\gcd(d,q+1+t)=1$. Then the family of binary sequences $\mathcal{S}=\{\mathbf{s}_{i}: 1\leq i\leq q^{d-1}-1\}$ constructed above is with length $q+1+t$, size $q^{d-1}-1$ and balance bounded by
         \begin{equation}\label{111111}
\Delta(\mathcal{S})\leq(d+1)\cdot\lfloor2\sqrt{q}\rfloor+|t|+d.
          \end{equation}
            \end{theorem}
      \begin{proof}
Consider the Kummer extension $E_{i}'=E(y)$ with the defining equation
\begin{displaymath}
    E'_{i}: y^{2}=z_{i}
\end{displaymath}
for each $1\leq i\leq q^{d-1}-1$. By Lemma III.2, we have $\mathfrak{g}_{E'_{i}}\leq 1+d$.

Denote by $T_{0}(\mathbf{s}_{i})$ the cardinality of the set $\mathcal{L}_{0}=\{0\leq j\leq q+t : \eta(z_{i}(P_{j}))=0\}$ and $T_{1}(\mathbf{s}_{i})$ the cardinality of the set $\mathcal{L}_{1}=\{0\leq j\leq q+t : \eta(z_{i}(P_{j}))=1\}$ for each $1\leq i\leq q^{d-1}-1$. By the relation $T_{0}(\mathbf{s}_{i})+T_{1}(\mathbf{s}_{i})=q+1+t$, we have 
\begin{displaymath}
\begin{split}
    |T_{0}(\mathbf{s}_{i})-T_{1}(\mathbf{s}_{i})|&=|2T_{1}(\mathbf{s}_{i})-(q+1+t)|\\
    &=|2T_{0}(\mathbf{s}_{i})-(q+1+t)|,
    \end{split}
    \end{displaymath}
then we only need to consider the cardinality $T_{0}(\mathbf{s}_{i})$.

By the definition of quadratic residue map $\eta$ and the properties of Kummer extensions, if $j\in\mathcal{L}_{0}$, then $P_{j}$ splits into 2 rational places in $E'_{i}$ or there is only one rational place lying over $P_{j}$ in $E'_{i}$. If $j\in\mathcal{L}_{1}$, then any place lying over $P_{j}$ has degree two.

                   On the other hand, we have $N(E'_{i})=2T_{0}(\mathbf{s}_{i})-|\mathcal{T}(z_{i})|$ for $1\leq i\leq q^{d-1}-1$ and $\mathcal{T}(z_{i})$ denotes the set of rational zeros of $z_{i}$. From the Serre bound, we have the estimation for the number of rational places over the function field $E'_{i}$ and it is given by
                  \begin{displaymath}
                  \begin{split}
                      |N(E'_{i})-q-1|&\leq \mathfrak{g}_{E'_{i}}\cdot\lfloor2\sqrt{q}\rfloor\\
                      &\leq (d+1)\cdot\lfloor2\sqrt{q}\rfloor.
                      \end{split}
                      \end{displaymath}

Then we have the following estimation:
\begin{displaymath}
\begin{split}
|T_{1}(\mathbf{s}_{i})-T_{0}(\mathbf{s}_{i})|&=\left|2T_{0}(\mathbf{s}_{i})-(q+1+t)\right|\\
    &=|N(E_{i}')+|\mathcal{T}(z_{i})|-(q+1+t)|\\
    &\leq |N(E'_{i})-q-1|+|t|+|\mathcal{T}(z_{i})|\\
    &\leq (d+1)\cdot\lfloor2\sqrt{q}\rfloor+|t|+d,
    \end{split}
    \end{displaymath}
for any $1\leq i\leq q^{d-1}-1$. Therefore we have the desired results.
\end{proof}

Based on the results above, we have the following upper bound for the correlation.
     \begin{theorem}\label{8}
         Let $q=p^{n}$ for odd prime $p$ and a positive integer $n$. Let $E/\mathbb{F}_{q}$ be a cyclic elliptic function field defined over $\mathbb{F}_{q}$ with $q+1+t$ rational places. Let $F$ be the fixed subfield of $E$ with respect to its translation group $T_{E}$ generated by $\sigma_{P}$. Let $Q$ be a place of $E$ with $\deg(Q)=d\geq 2$ such that $\gcd(d,q+1+t)=1$. Then the family of binary sequences $\mathcal{S}=\{\mathbf{s}_{i}: 1\leq i\leq q^{d-1}-1\}$ constructed above is with length $q+1+t$, size $q^{d-1}-1$ and the correlation upper bounded by
         \begin{displaymath}
         \begin{split}
             \mathrm{Cor}(\mathcal{S})
             &\leq (2d+1)\cdot\lfloor2\sqrt{q}\rfloor+|t|+2d.
             \end{split}
             \end{displaymath}
     \end{theorem}
         \begin{proof}
             Let $\tau=\sigma_{P}$ be the generator of the translation group $T_{E}$. We begin by computing the autocorrelation of the family of binary sequences $\mathcal{S}=\{\mathbf{s}_{i} : 1\leq i\leq q^{d-1}-1\}$. 
             
             According to the definition, for each $1\leq i\leq q^{d-1}-1$, the autocorrelation of $\mathbf{s}_{i}$ at delay $u$ with $1\leq u\leq q+t$ is given by
             \begin{displaymath}
             \begin{split}
                 \mathrm{A}_{u}(\mathbf{s}_{i})&=\sum^{q+t}_{j=0}(-1)^{s_{i,j}+s_{i,j+u}}
                 =\sum^{q+t}_{j=0}(-1)^{\eta(z_{i}(P_{j}))+\eta(z_{i}(P_{j+u}))}\\
                 &=\sum^{q+t}_{j=0}(-1)^{\eta(z_{i}(P_{j})+\eta(\tau^{-u}(z_{i})(P_{j}))}\\
                  &=\sum^{q+t}_{j=0}(-1)^{\eta(z_{i}\tau^{-u}(z_{i})(P_{j}))}.\\
                  \end{split}
                 \end{displaymath}
                 
                 The third equality follows from Lemma \ref{1}. Since $z_{i}\in V\setminus\{0\}$ and $1\leq u\leq q+t$, the pole divisor of $z_{i}\cdot \tau^{-u}(z_{i})$ is given by $(z_{i}\tau^{-u}(z_{i}))_{\infty}=Q+\tau^{-u}(Q)$ by Lemma \ref{2}. Consider the Kummer extension $E_{i}=E(y)$ with defining equation
                 \begin{displaymath}
                   E_{i}:  y^{2}=z_{i}\cdot\tau^{-u}(z_{i}).
                 \end{displaymath}
                 
                 Denote by $\mathcal{Z}(z_{i})$ the zeros of the rational function $z_{i}$ for $1\leq i\leq q^{d-1}-1$. By Lemma \ref{1} and Lemma III.2, we know that the places $Q$, $\tau^{-u}(Q)$, and zeros $\mathcal{Z}(z_{i})$ are all possible ramified places in $E_{i}/E$. From the Hurwitz genus formula, we have
                 \begin{displaymath}
                     2\mathfrak{g}_{E_{i}}-2\leq 2(2\mathfrak{g}_{E}-2)+\deg(Q)+\deg(\tau^{-u}(Q))+2\sum_{P\in\mathcal{Z}(z_{i})}\deg(P).
                 \end{displaymath}
                 Then the genus of $E_{i}$ satisfies
                 \begin{displaymath}
                     \mathfrak{g}_{E_{i}}\leq 2d+1
                     \end{displaymath}
                  for each $1\leq i\leq q^{d-1}-1$.

                  Denote by $N_{0}$ the cardinality of the set $\mathcal{S}_{0}=\{0\leq j\leq q+t : \eta((z_{i}\tau^{-u}(z_{i})(P_{j}))=0    \}$ and $N_{1}$ the cardinality of the set $\mathcal{S}_{1}=\{0\leq j\leq q+t : \eta((z_{i}\tau^{-u}(z_{i})(P_{j}))=1    \}$. It is clear that 
                  \begin{displaymath}
                      N_{0}+N_{1}=q+1+t.
                  \end{displaymath}

By the definition of the quadratic residue map $\eta$ and the properties of Kummer extensions, if $j\in\mathcal{S}_{0}$, then $P_{j}$ splits into 2 rational places in $E_{i}$ or there is only one rational place lying over $P_{j}$ in $E_{i}$. If $j\in\mathcal{S}_{1}$, then any place lying over $P_{j}$ has degree two.

                  Denote by $\mathcal{T}(z_{i})$ the set of rational zeros of $z_{i}$, then the number of the rational points of $E_{i}$ can be given by $N(E_{i})=2N_{0}-|\mathcal{T}(z_{i}\tau^{-u}(z_{i}))|$ for each $1\leq i\leq q^{d-1}-1$. 
                  
                  From the Serre bound, we have
                  \begin{displaymath}
                  \begin{split}
                      |N(E_{i})-q-1|&\leq \mathfrak{g}_{E_{i}}\cdot\lfloor2\sqrt{q}\rfloor\\
                      &\leq \left(2d+1\right) \cdot\lfloor2\sqrt{q}\rfloor.
                      \end{split}
                      \end{displaymath}
                   Therefore, we obtain the upper bound for the auto-correlation
                   \begin{displaymath}
                   \begin{split}
                       |\mathrm{A}_{u}(\mathbf{s}_{i})|&=|N_{0}-N_{1}|=|2N_{0}-q-1-t|\\
                       &\leq|N(E_{i})-q-1|+|\mathcal{T}(z_{i}\tau^{-u}(z_{i}))|+|t| \\
                       &\leq \left(2d+1\right) \cdot\lfloor2\sqrt{q}\rfloor+|t|+ |\mathcal{T}(z_{i}\tau^{-u}(z_{i}))|\\
                       &\leq  \left(2d+1\right) \cdot\lfloor2\sqrt{q}\rfloor+|t|+ 2|\mathcal{T}(z_{i})|\ (\text{by Lemma \ref{1}})\\                    &\leq  \left(2d+1\right) \cdot\lfloor2\sqrt{q}\rfloor+|t|+ 2|\mathcal{Z}(z_{i})| \\                        &\leq  \left(2d+1\right) \cdot\lfloor2\sqrt{q}\rfloor+|t|+2d.
                       \end{split}
                       \end{displaymath}
                            The last inequality holds from the fact that $z_{i}\in\mathcal{L}(Q)$, thus we have $|\mathcal{T}(z_{i})|\leq|\mathcal{Z}(z_{i})|\leq d$ for each $1\leq i\leq q^{d-1}-1 $. 
                  
                  Now we calculate the cross-correlation for the family of sequences. For two arbitrary distinct sequences $\mathbf{s}_{i}$ and $\mathbf{s}_{j}$ in $\mathcal{S}$ for $1\leq i\neq j\leq q^{d-1}-1$, the cross-correlation of $\mathbf{s}_{i}$ and $\mathbf{s}_{j}$ at delay $u$ with $0\leq u\leq q+t$ is given by
                  \begin{displaymath}
                  \begin{split}
                      \mathrm{C}_{u}(\mathbf{s}_{i},\mathbf{s}_{j})&=\sum^{q+t}_{k=0}(-1)^{\eta(z_{i}(P_{k}))+\eta(z_{j+u}(P_{k}))}\\
                      &=\sum^{q+t}_{k=0}(-1)^{\eta(z_{i}\tau^{-u}(z_{j})(P_{k}))}.
                      \end{split}
                  \end{displaymath}
    
    Consider the Kummer extension $E_{i,j}=E(y)$ with defining equation
    \begin{displaymath}
     E_{i,j}:  y^{2}=z_{i}\cdot\tau^{-u}(z_{j}).
        \end{displaymath}

If $1\leq u\leq q+t$, then we have $(z_{i}\tau^{-u}(z_{j}))_{\infty}=Q+\tau^{-u}(Q)$ by Lemma \ref{1}. From the Hurwitz genus formula, we have
\begin{displaymath}
\begin{split}
    2\mathfrak{g}_{E_{i,j}}-2&\leq 2(2\mathfrak{g}_{E}-2)+\deg(Q)+\deg(\tau^{-u}(Q))+\sum_{P\in\mathcal{Z}(z_{i})}\deg(P)+\sum_{P\in\mathcal{Z}(z_{j})}\deg(P)\\
    &\leq\deg(Q)+\deg(\tau^{-u}(Q))+2d=4d.
    \end{split}
\end{displaymath}

Hence, the genus of $E_{i,j}$ satisfies
\begin{displaymath}
    \mathfrak{g}_{E_{i,j}}\leq 2d+1   
    \end{displaymath}
for $1\leq i\neq j\leq q^{d-1}-1$. If $u=0$, then we have $(z_{i}\tau^{-u}(z_{j}))_{\infty}=(z_{i}z_{j})_{\infty}=2Q$ by Lemma \ref{2} and the genus of $E_{i,j}$ is also upper bounded by $\mathfrak{g}_{E_{i,j}}\leq 2d+1$.

Let $N'_{0}$ denote the cardinality of the set $\mathcal{S}_{0}'=\{0\leq k\leq q+t : \eta(z_{i}\tau^{-u}(z_{j})(P_{k}))=0  \}$ and $N'_{1}$ denote the cardinality of the set $\mathcal{S}'_{1}=\{0\leq k\leq q+t : \eta(z_{i}\tau^{-u}(z_{j})(P_{k}))=1  \}$. It is clear that
\begin{displaymath}
    N'_{0}+N'_{1}=q+1+t.
\end{displaymath}
Also by the definition of quadratic residue map $\eta$ and the properties of Kummer extensions, if $j\in\mathcal{S}'_{0}$, then $P_{j}$ splits into 2 rational places in $E_{i,j}$ or there is only one rational place lying over $P_{j}$ in $E_{i,j}$. If $j\in\mathcal{S}'_{1}$, then any place lying over $P_{j}$ has degree two.

                   On the other hand, $N(E_{i,j})=2N'_{0}-|\mathcal{T}(z_{i}\tau^{-u}(z_{j}))|$ for $1\leq u\leq q+t$. From the Serre bound, we have the estimation for the number of the rational places on the function field $E_{i,j}$
                  \begin{displaymath}
                  \begin{split}
                      |N(E_{i,j})-q-1|&\leq \mathfrak{g}_{E_{i,j}}\cdot\lfloor2\sqrt{q}\rfloor\\
                      &\leq (2d+1)\cdot\lfloor2\sqrt{q}\rfloor.
                      \end{split}
                      \end{displaymath}
                   Therefore, we obtain the following upper bound for the cross-correlation
                   \begin{displaymath}
                   \begin{split}
                       |\mathrm{C}_{u}(\mathbf{s}_{i},\mathbf{s}_{j})|&=|N'_{0}-N'_{1}|=|2N'_{0}-q-1-t|\\
                       &\leq|N(E_{i,j})-q-1|+|\mathcal{T}(z_{i}\tau^{-u}(z_{j}))|+|t| \\
                       &\leq \left(2d+1\right) \cdot\lfloor2\sqrt{q}\rfloor+|t|+ |\mathcal{T}(z_{i}\tau^{-u}(z_{j}))|\\
                       &\leq  \left(2d+1\right) \cdot\lfloor2\sqrt{q}\rfloor+|t|+ |\mathcal{T}(z_{i})|+|\mathcal{T}(z_{j})|\\                    
                       &\leq  \left(2d+1\right) \cdot\lfloor2\sqrt{q}\rfloor+|t|+2d
                       \end{split}
                       \end{displaymath}
for all $0\leq u\leq q+t$.

Combining the results of  $|\mathrm{A}_{u}(\mathbf{s}_{i})|$      and $|\mathrm{C}_{u}(\mathbf{s}_{i},\mathbf{s}_{j})|$ for $1\leq i\neq j\leq q^{d-1}-1$, we have the desired results.
 \end{proof}

\begin{corollary}\label{c0}
$1)$ Let $n$ be an odd positive integer, $q=3^{n}$ and $t=\pm 3^{(n+1)/2}$. Then there exists a family of binary sequences with length $q+1+t$, size $q-1$, balance upper bounded by $\Delta(\mathcal{S})\leq 3\cdot\lfloor2\cdot 3^{n/2}\rfloor+3^{(n+1)/2}+2$ and the correlation upper-bounded by $\mathrm{Cor}(\mathcal{S})\leq 5\cdot\lfloor2\cdot 3^{n/2}\rfloor+3^{(n+1)/2}+4$.

    $2)$ Let $n$ be an odd positive integer, $q=3^{n}$ and $t=-1$. Then there exists a family of binary sequences with length $q$, size $q-1$, balance upper bounded by $\Delta(\mathcal{S})\leq 3\cdot\lfloor2\cdot 3^{n/2}\rfloor+2$ and the correlation upper-bounded by $\mathrm{Cor}(\mathcal{S})\leq 5\cdot\lfloor2\cdot 3^{n/2}\rfloor+4$.
    
   $3)$ Let $q=5^{n}$ for a positive integer $n$. Let $t=\pm 5^{n/2}$ and $n$ be even. Then there exists a family of binary sequences with length $q+1+t$, size $q-1$, balance upper bounded by $\Delta(\mathcal{S})\leq 3\cdot\lfloor2\cdot 5^{n/2}\rfloor+|t|+2=7\cdot 5^{n/2}+2$ and the correlation upper bounded by $\mathrm{Cor}(\mathcal{S})\leq 5\cdot\lfloor2\cdot 5^{n/2}\rfloor+|t|+4=11\cdot 5^{n/2}+4$.
\end{corollary}
    \begin{proof}
   $1)$ For $t=\pm 3^{(n+1)/2}$, there is a cyclic elliptic function field with $q+1+t$ rational places according to Lemma~\ref{123}. By Lemma~\ref{l6}, there exists a place $Q$ in $E^{T_{E}}$ of degree $\deg(Q)=2$ that splits completely in $E$ as $\gcd(2,3^{n}+1+t)=1$. Thus, the corollary follows from Theorems  \ref{balance} and \ref{8}. 
   
        $2)$  For $t=-1$, there is a cyclic elliptic function field with $q+1+t=q$ rational places according to Lemma~\ref{123}. By Lemma~\ref{l6}, there exists a place $Q$ in $E^{T_{E}}$ of degree $\deg(Q)=2$ that splits completely in $E$ as $\gcd(2,3^{n})=1$.  Thus, the corollary follows from Theorems  \ref{balance} and \ref{8}.     
        
        $3)$ For $t=\pm 5^{n/2}$, there is a cyclic elliptic function field with $q+1+t$ rational places according to Lemma~\ref{123}. By Lemma~\ref{l6}, there exists a place $Q$ in $E^{T_{E}}$ of degree $\deg(Q)=2$ that splits completely in $E$ for the fact that $\gcd(2,5^{n}+1+t)=1$.  Thus, the corollary follows from Theorems  \ref{balance} and \ref{8} as well.     
    \end{proof}

\subsection{The Linear Complexity}
   In this subsection, we shall show that the family of binary sequences $\mathcal{S}=\{\mathbf{s}_{i} : 1\leq i\leq q^{d-1}-1\}$ also has a large linear complexity.
   \begin{theorem}
       The linear complexity of the family of binary sequences $\mathcal{S}=\{\mathbf{s}_{i} : 1\leq i\leq q^{d-1}-1\}$ satisfies
       \begin{displaymath}
           \mathrm{LC}(\mathcal{S})\geq\frac{q+1+2t-d-(d+1)\cdot\lfloor2\sqrt{q}\rfloor}{d+d\cdot\lfloor2\sqrt{q}\rfloor}.               \label{222}
           \end{displaymath}     \end{theorem}
   \begin{proof}
           We only need to prove that
           \begin{displaymath}
               \ell(\mathbf{s}_{i})\geq\frac{q+1+2t-d-(d+1)\cdot\lfloor2\sqrt{q}\rfloor}{d+d\cdot\lfloor2\sqrt{q}\rfloor}               
               \end{displaymath}
           for each sequence $\mathbf{s}_{i}\in\mathcal{S}, 1\leq i\leq q^{d-1}-1$.

  Denote by $\ell=\ell(\mathbf{s}_{i})$. Then there exist $\ell+1$ elements with $\lambda_{0}=\lambda_{\ell}=1$ and $\lambda_{1},\cdots,\lambda_{\ell-1}\in\mathbb{F}_{2}$ such that
           \begin{displaymath}
               \sum^{\ell}_{j=0}\lambda_{j}\cdot\eta(z_{i}(P_{j+u}))=0
           \end{displaymath}
           for any integer $u\geq 0$. Let $\tau=\sigma_{P}$ be a generator of the translation group $T_{E}$. By the definition of linear complexity, we have 
           \begin{displaymath}
           \begin{split}
               \sum^{\ell}_{j=0}\lambda_{j}\cdot\eta(z_{i}(P_{j+u}))
           &=\eta\left(\prod^{\ell}_{j=0}z_{i}^{\lambda_{j}}(P_{j+u})\right)\\
               &=\eta\left(\prod^{\ell}_{j=0}(z_{i})^{\lambda_{j}}(\tau^{j}(P_{u}))\right)\\               &=\eta\left(\prod^{\ell}_{j=0}(\tau^{-j}(z_{i}))^{\lambda_{j}}(P_{u})\right)\\
               &=\eta\left(w_{i}(P_{u})\right)\\          
               \end{split}
           \end{displaymath}
where $w_{i}=\prod^{\ell}_{j=0}(\tau^{-j}(z_{i}))^{\lambda_{j}}$. Hence, we obtain the vector 
\begin{displaymath}
    (\eta(w_{i}(P_{1})),\eta(w_{i}(P_{2})),\cdots,\eta(w_{i}(P_{q+1+t})))=\mathbf{0}.
\end{displaymath}
Consider the Kummer extension $L_{i}/E$ with defining equation
\begin{displaymath}
   L_{i}: y^{2}=\prod^{\ell}_{j=0}(\tau^{-j}(z_{i}))^{\lambda_{j}}=w_{i}.
\end{displaymath}
By Lemma \ref{2} and Lemma III.2, we have $(w_{i})_{\infty}=\sum^{\ell}_{j=0}\lambda_{j}\tau^{-j}(Q)$ and the genus can also be calculated by
\begin{displaymath}
\begin{split}
    2\mathfrak{g}_{L_{i}}-2&=2\cdot(2\mathfrak{g}_{E}-2)+\deg(\mathrm{Diff}(L_{i}/E))\\
    &=\deg(\mathrm{Diff}(L_{i}/F))\\
    &\leq (\ell+1)\left(d+\sum_{P\in\mathcal{Z}(z_{j})}\deg(P)\right) (\text{for some}\ 0\leq j\leq \ell)\\
    &\leq (\ell+1)\cdot 2d.
    \end{split}
\end{displaymath}
Then we have an upper bound for the genus $\mathfrak{g}_{L_{i}}$:
\begin{displaymath}
    \mathfrak{g}_{L_{i}}\leq(\ell+1)d+1.
    \end{displaymath}
It follows that 
\begin{displaymath}
   2(q+1+t)-|\mathcal{T}(w_{i})| \leq N(L_{i})\leq q+1+\left[(\ell+1)d+1\right]\cdot\lfloor2\sqrt{q}\rfloor,
\end{displaymath}
 where $\mathcal{T}(w_{i})$ denotes the set of rational zeros of $w_{i}$ for $1\leq i\leq q^{d-1}-1$. Since $|\mathcal{T}(w_{i})|\leq |\mathcal{Z}(w_{i})|$, we have
\begin{displaymath}
   2(q+1+t)-|\mathcal{Z}(w_{i})| \leq N(L_{i})\leq q+1+\left[(\ell+1)d+1\right]\cdot\lfloor2\sqrt{q}\rfloor.
\end{displaymath}
Furthermore, we have the following inequalities
\begin{displaymath}
\begin{split}
    2(q+1+t)&\leq|\mathcal{Z}(w_{i})|+q+1+\left[(\ell+1)d+1\right]\cdot\lfloor2\sqrt{q}\rfloor\\
    &\leq (\ell+1)\cdot\left(d+d\cdot\lfloor2\sqrt{q}\rfloor\right)+q+1+\lfloor2\sqrt{q}\rfloor.
    \end{split}
    \end{displaymath}
Then we have
\begin{displaymath}
    \ell(\mathbf{s}_{i})
    \geq\frac{q+1+2t-d-(d+1)\cdot\lfloor2\sqrt{q}\rfloor}{d+d\cdot\lfloor2\sqrt{q}\rfloor}\\    
\end{displaymath}
for all $1\leq i\leq q^{d-1}-1$.
       \end{proof}
     
\subsection{Comparison}
    
   In this subsection, we compare our construction with the binary sequences proposed by Jin \textit{et.al.} in \cite{3}. As we have mentioned in Section \ref{sec:1}, the main tools are taken from the theory of Kummer extensions while the main techniques are Artin-Schreier extensions in \cite{3}. Now we discuss the main differences in detail.
     
    The problems of computing correlation are transformed into estimating the number of rational places of certain Kummer extensions $E_{j,k}=E(y)$ where $y^{2}=z_{j}\cdot\tau^{-u}(z_{k})$ for $1\leq j,k\leq q^{d-1}-1, 0\leq u\leq q+t$ by the cyclic structure of the translation group $T_{E}$. Because our construction is based on the cyclic elliptic function fields over finite fields with odd characteristics, our construction can generate binary sequences with more flexible parameters.
     
 By contrast, the binary sequences~\cite{3} constructed from the finite field $\mathbb{F}_{2^{n}}$ utilize the trace map $\mathrm{Tr}:\mathbb{F}_{2^{n}}\to\mathbb{F}_{2}$ where the problem of computing the correlation is transformed into estimating the number of rational places of the Artin-Schreier extensions $E'_{j,k}=E(y)$ where $y^{2}+y=z_{j}+\tau^{-u}(z_{k})$ for $1\leq j,k\leq q^{d-1}-1, 0\leq u\leq q+t$.
 
In order to obtain a lower bound for the linear complexity, we need to estimate the rational places of the Kummer extensions 
 $ L_{i}: y^{2}=\prod^{\ell}_{j=0}(\tau^{-j}(z_{i}))^{\lambda_{j}}$ for $1\leq i\leq q^{d-1}-1$ in our construction. Instead, Jin \textit{et al.}~\cite{3} need to estimate the rational places of the Artin-Schreier extensions $ L_{i}: y^{2}+y=\sum^{\ell}_{j=0}(\tau^{-j}(z_{i}))^{\lambda_{j}}$ for $1\leq i\leq q^{d-1}-1$.

Note that the ramification behaviors are more complex in Kummer extension. Both the zero places and pole places of $z_{j}\cdot\tau^{-u}(z_{k})$ are ramified in the Kummer extensions while only the pole places of $z_{j}+\tau^{-u}(z_{k})$ are ramified in the Artin-Schreier extensions for $1\leq j,k\leq q^{d-1}-1, 0\leq u\leq q+t$, which induces the differences on the bounds.

    \subsection{Numerical Results}
  In this paper, we provide a novel method to construct a family of binary sequences with favorable balance, low correlation and large linear complexity by utilizing cyclic elliptic function fields over finite fields with odd characteristics. In this subsection, we provide some numerical results of our construction of binary sequences. 
  It turns out that the actual balance and correlation of such new binary sequences are significantly lower than the theoretical upper bound provided in Theorem \ref{8}. 
  
  In the following, we consider the cases for the place $\deg(Q)=d=2$ and the cyclic elliptic function fields over finite fields with characteristic $p=3$.

 First, we consider the balance. Let $q=3^{n}$. It is also obvious that $\gcd(2,q)=1$. By Corollary \ref{c0}, there exists a cyclic elliptic function field $E/\mathbb{F}_{q}$ with $q$ rational places. 
By Theorem \ref{balance}, we can construct a family of binary sequences with size $q-1$, length $q$ and balance $\Delta(\mathcal{S})$ upper bounded by $3\cdot\lfloor2\sqrt{q}\rfloor+2$.    We list some numerical results in Table \ref{tab:5} and it turns out that the actual value of $\Delta(\mathcal{S})$ is much smaller than the bound given in Theorem III.5.
\begin{table}[htbp]
        \centering
         \caption{$\substack{\text{Balance and Correlation of the first construction of sequences via the cyclic elliptic function fields}\\ \text{over finite fields with characteristic $p=3$}}$}    \begin{tabular}{|c|c|c|c|c|c|c|}
        \toprule
          Field Size& Length    & Family Size& Bound of Balance & Actual Balance&Bound of Correlation&Actual Correlation\\
          \midrule
          81&81&80& 56&17&94&43\\
          \hline
           243&243   & 242&95&55 &159&73 \\
           \hline
           729&729 &728&164 &65&274&131\\
           \hline
           2187&2187&2186&281&97&496&274\\
           \bottomrule
        \end{tabular}
        \label{tab:5}
    \end{table}
  
Now we consider the actual correlation. By the results in Theorem \ref{8}, the correlation is upper bounded by $5\cdot\lfloor2\sqrt{q}\rfloor+4$. We also list some results calculated by SageMath~\cite{sage} for this case in Table \ref{tab:5}. It turns out that the actual correlation is much lower than the theoretical bound.

Now consider another construction of the binary sequences from the cyclic elliptic function fields over $\mathbb{F}_{q}$. Let $t=\sqrt{q}$ if $n$ is even, or $\sqrt{3q}$ if $n$ is odd. By Lemma \ref{123}, we can construct a family of binary sequences with family size $q-1$, length $q+1+t$ and the correlation upper bounded by $5\cdot\lfloor2\sqrt{q}\rfloor+t+4$. We list some numerical results in Table \ref{tab:4}.
 \begin{table}[htbp]
        \centering
         \caption{$\substack{\text{Parameters of the second construction of sequences via the cyclic elliptic function fields}\\ \text{over finite fields with characteristic $p=3$}}$}    \begin{tabular}{|c|c|c|c|c|}
        \toprule
          Field Size   & Length & Family Size & Bound of Correlation &Actual Correlation\\
          \midrule
          81&91&80&103&46\\
          \hline
           243  &271 & 242&186& 78 \\
           \hline
           729&757 &728&302 &148\\
           \bottomrule
        \end{tabular}
        \label{tab:4}
    \end{table}    
For the detailed programming code, we list it on the webpage https://github.com/lxfhah-byte/binary-sequences-odd.


\section{ Conclusions and Future Work}
\label{sec:4}
 In this paper, we extend the construction of binary sequences given by Jin \textit{et al.} in~\cite{3} via the cyclic elliptic function fields with characteristic $2$ to general odd characteristics by applying the quadratic residue map $\eta$ and the theory of Kummer extensions of function fields. Since our construction is over elliptic function fields with odd characteristics, it turns out that our construction has flexible parameters. The balance, correlation and linear complexity of the binary sequences are investigated.  Upper bounds for the balance, the correlation and a lower bound for the linear complexity are obtained. 
 From the numerical experiments, it is noticed that the real value of balance or correlation is much better than the upper bound given in Theorem~\ref{balance} or Theorem~\ref{8} respectively. For future work, it is interesting to 1) give better bounds for correlations and linear complexities of the sequences constructed in~\cite{3} and this paper by using other deep mathematical tools and 2) generalize the constructions of binary sequences by utilizing general algebraic function fields.

\end{document}